\documentclass[lettersize,journal]{IEEEtran}
\usepackage[T1]{fontenc} % optional
\usepackage{tikz}
\usepackage{amssymb}
\usepackage{amsthm}
\usepackage{amsmath}
\usepackage{balance}
\usepackage{subcaption}
\usepackage{caption}
\usepackage{pifont}
\usepackage{enumitem}
\usepackage{booktabs}
\usepackage{multirow}
\usepackage{pgfplots}
\usepackage{graphicx}
\usepackage{glossaries}
\usepackage{array}
\usepackage{comment}
\usepackage{multirow}
\usepackage{makecell}
\usepackage{cite}
\usetikzlibrary{shapes}
\usepackage{float}
\usepackage{stfloats}
\usepackage[capitalise]{cleveref}
\usepackage{verbatim}
\usepackage{textcomp}
\usepackage{titlesec}

\newcommand{\cmark}{\ding{51}}%
\newcommand{\bcmark}{\ding{52}}%
\newcommand{\xmark}{\ding{55}}%
\newcommand\mytextsf{\bfseries\sffamily\fontsize{9pt}{9pt}\selectfont}

\newtheorem{theorem}{Theorem}
\newtheorem{proposition}{Proposition}

\newtheorem{lemma}[theorem]{Lemma}

\newtheoremstyle{myremark}% name of the style to be used
  {\topsep}% measure of space to leave above the theorem. E.g.: 3pt
  {\topsep}% measure of space to leave below the theorem. E.g.: 3pt
  {\itshape}% name of font to use in the body of the theorem
  {0pt}% measure of space to indent
  {\scshape}% name of head font
  {.}% punctuation between head and body
  { }% space after theorem head; " " = normal interword space
  {}% standard setting

\theoremstyle{myremark}
\newtheorem{remark}{Remark}

\makeatletter
\let\save@mathaccent\mathaccent
\newcommand*\if@single[3]{%
  \setbox0\hbox{${\mathaccent"0362{#1}}^H$}%
  \setbox2\hbox{${\mathaccent"0362{\kern0pt#1}}^H$}%
  \ifdim\ht0=\ht2 #3\else #2\fi
  }
\newcommand*\rel@kern[1]{\kern#1\dimexpr\macc@kerna}
\newcommand*\widebar[1]{\@ifnextchar^{{\wide@bar{#1}{0}}}{\wide@bar{#1}{1}}}
\newcommand*\wide@bar[2]{\if@single{#1}{\wide@bar@{#1}{#2}{1}}{\wide@bar@{#1}{#2}{2}}}
\newcommand*\wide@bar@[3]{%
  \begingroup
  \def\mathaccent##1##2{%
    \let\mathaccent\save@mathaccent
    \if#32 \let\macc@nucleus\first@char \fi
    \setbox\z@\hbox{$\macc@style{\macc@nucleus}_{}$}%
    \setbox\tw@\hbox{$\macc@style{\macc@nucleus}{}_{}$}%
    \dimen@\wd\tw@
    \advance\dimen@-\wd\z@
    \divide\dimen@ 3
    \@tempdima\wd\tw@
    \advance\@tempdima-\scriptspace
    \divide\@tempdima 10
    \advance\dimen@-\@tempdima
    \ifdim\dimen@>\z@ \dimen@0pt\fi
    \rel@kern{0.6}\kern-\dimen@
    \if#31
      \overline{\rel@kern{-0.6}\kern\dimen@\macc@nucleus\rel@kern{0.4}\kern\dimen@}%
      \advance\dimen@0.4\dimexpr\macc@kerna
      \let\final@kern#2%
      \ifdim\dimen@<\z@ \let\final@kern1\fi
      \if\final@kern1 \kern-\dimen@\fi
    \else
      \overline{\rel@kern{-0.6}\kern\dimen@#1}%
    \fi
  }%
  \macc@depth\@ne
  \let\math@bgroup\@empty \let\math@egroup\macc@set@skewchar
  \mathsurround\z@ \frozen@everymath{\mathgroup\macc@group\relax}%
  \macc@set@skewchar\relax
  \let\mathaccentV\macc@nested@a
  \if#31
    \macc@nested@a\relax111{#1}%
  \else
    \def\gobble@till@marker##1\endmarker{}%
    \futurelet\first@char\gobble@till@marker#1\endmarker
    \ifcat\noexpand\first@char A\else
      \def\first@char{}%
    \fi
    \macc@nested@a\relax111{\first@char}%
  \fi
  \endgroup
}
\makeatother

\makeglossaries

\newacronym{ISAC}{ISAC}{integrated sensing and commmunications}
\newacronym{BS}{BS}{base station}
\newacronym{RF}{RF}{radio-frequency}
\newacronym{DAC}{DAC}{digital-to-analog converter}
\newacronym{IRS}{IRS}{intelligent reflecting surface}
\newacronym{PAPR}{PAPR}{peak-to-average power ratio}

\newacronym{EHF}{EHF}{extremely high frequency}
\newacronym{EHFs}{EHFs}{extremely high frequencies}

\newacronym{AWGN}{AWGN}{additive white Gaussian noise}
\newacronym{SNR}{SNR}{signal-to-noise ratio}
\newacronym{SDR}{SDR}{semidefinite relaxation}
\newacronym{SDP}{SDP}{semidefinite programming}
\newacronym{SCA}{SCA}{successive convex approximation}
 
\newacronym{MILP}{MILP}{mixed-integer linear program} 
\newacronym{MINLP}{MINLP}{mixed-integer nonlinear program} 

\newacronym{AoA}{AoA}{angle of arrival}
\newacronym{AoD}{AoD}{angle of departure}
\newacronym{DPG}{DPG}{directional power gain}

\newacronym{LoS}{LoS}{line-of-sight}
\newacronym{NLoS}{NLoS}{non-LoS}

\newacronym{LHS}{LHS}{left-hand-side}
\newacronym{RHS}{RHS}{right-hand-side}

\newacronym{ES}{ES}{exhaustive search}
\newacronym{BnC}{BnC}{branch-and-cut}

\begin{document}

% !TeX root = ../main.tex

\title{\huge Hierarchical Functionality Prioritization in Multicast ISAC: Optimal Admission Control and Discrete-Phase Beamforming}

\author{
\IEEEauthorblockN{Luis F. Abanto-Leon and 
				  Setareh Maghsudi}
}

%\author{
%\IEEEauthorblockN{Luis F. Abanto-Leon and 
%				  Setareh Maghsudi} \\ 
%\IEEEauthorblockA{Ruhr University Bochum, Germany \\
%\{luis.abantoleon, setareh.maghsudi\}@ruhr-uni-bochum.de} 
%}

%\author{
%\IEEEauthorblockN{Author 1 and 
%				  Author 2}
%}

\maketitle

% !TeX root = ../main.tex

\begin{abstract}
We investigate the joint admission control and discrete-phase multicast beamforming design for \gls{ISAC} systems, where sensing and communications functionalities have different hierarchies. Specifically, the \gls{ISAC} system first allocates resources to the higher-hierarchy functionality and opportunistically uses the remaining resources to support the lower-hierarchy one. This resource allocation problem is a nonconvex \gls{MINLP}. We propose an exact \gls{MILP} reformulation, leading to a globally optimal solution. In addition, we implemented three baselines for comparison, which our proposed method outperforms by more than $ 39 \% $.
\end{abstract}

% !TeX root = ../main.tex


% !TeX root = ../main.tex

\glsresetall
\section{Introduction} \label{sec:introduction}

\Gls{ISAC} is a disruptive advancement in wireless technology in which sensing and communications share the same radio resources, e.g., infrastructure, spectrum, waveform, to enhance radio resource utilization, reduce costs, and simplify system complexity \cite{liu2020:joint-radar-communication-design-applications-state-art-road-ahead}.

Sensing at high frequencies is appealing since the shorter wavelengths enable finer resolution \cite{mao2022:waveform-design-joint-sensing-communications-millimeter-wave-low-terahertz-bands}. These frequencies suffer severe path loss, which beamforming can alleviate. Highly versatile digital beamformers are expensive to manufacture for such high frequencies. Hence, analog beamformers lead the initial stages of \gls{ISAC} systems operating at these frequencies. 

Analog beamformers can be designed with continuous or discrete phases. The state-of-the-art literature features beamforming designs with both phase types, but most works focused on continuous phases, e.g., \cite{cao2022:pareto-optimal-analog-beamforming-design-integrated-mimo-radar-communication,  nguyen2023:multiuser-mimo-wideband-joint-communications-sensing-system-subcarrier-allocation, lyu2024:crb-minimization-ris-aided-mmwave-integrated-sensing-communications}, while a few accounted for discrete phases, e.g., \cite{xu2023:joint-antenna-selection-beamforming-integrated-automotive-radar-sensing-communications-quantized-double-phase-shifters}. The latter are of immense practical interest, as they reduce system complexity and costs. To date, however, only suboptimal beamforming designs exist for \gls{ISAC} systems that utilize discrete phases.

Another characteristic of analog beamformers is their single \gls{RF} chain, which supports one signal stream, making them well-suited for multicasting scenarios, such as broadcasting live sports or concerts to several subscribed users simultaneously. Multicast beamforming has been well investigated in non-\gls{ISAC} systems, e.g., \cite{demir2015:optimum-discrete-phase-only-multicast-beamforming-joint-antenna-user-selection-cognitive-radio-networks, abanto2023:radiorchestra-proactive-management-millimeter-wave-self-backhauled-small-cells-joint-optimization-beamforming-user-association-rate-selection-admission-control}, but rarely in \gls{ISAC} systems, with only a few studies addressing the topic, e.g., \cite{ ren2024:fundamental-crb-rate-tradeoff-multi-antenna-isac-systems-information-multicasting-multi-target-sensing}. Yet, none of such studies accounted for constant-modulus discrete phases. Particularly, multicasting and ISAC could play a key role in live events where drones are often used for aerial filming. Thus, ISAC could enable drone tracking while supporting efficient content dissemination to users.

%Multicasting and ISAC can be key in scenarios such as live concerts where drones are often employed for aerial filming. Thus, ISAC can enable tracking of drones while supporting efficient content dissemination to users.

%In non-\gls{ISAC} systems, admission control plays a crucial role in preventing resource allocation infeasibility when radio resources are limited, allowing to flexibilize constraints, such as serving only a selected subset of users \cite{abanto2023:radiorchestra-proactive-management-millimeter-wave-self-backhauled-small-cells-joint-optimization-beamforming-user-association-rate-selection-admission-control, matskani2009:efficient-batch-adaptive-approximation-algorithms-joint-multicast-beamforming-admission-control}, thereby enhancing resource utilization. In light of its advantages, incorporating admission control into \gls{ISAC} systems holds significant promise. However, despite its potential, this aspect has been overlooked in \gls{ISAC} contexts.

In non-\gls{ISAC} systems, admission control is crucial in preventing resource allocation infeasibility, especially when radio resources are limited, allowing to serve only a selected subset of users \cite{abanto2023:radiorchestra-proactive-management-millimeter-wave-self-backhauled-small-cells-joint-optimization-beamforming-user-association-rate-selection-admission-control, matskani2009:efficient-batch-adaptive-approximation-algorithms-joint-multicast-beamforming-admission-control}, thereby enhancing resource utilization. In light of its advantages, incorporating admission control into \gls{ISAC} systems holds significant promise. However, despite its potential, this aspect has been overlooked in \gls{ISAC} contexts.

Moreover, angular positions of targets may not be known precisely due to factors such as motion. Thus, accounting for this aspect in the resource allocation design can help mitigate potential performance degradation in sensing, a crucial aspect explored in only a few studies, such as \cite{xu2018:robust-resource-allocation-uav-systems-uav-jittering-user-location-uncertainty}.

In \gls{ISAC} systems, one functionality may be more critical than the other \cite{cui2023:integrated-sensing-communications-background-applications}. Particularly, this view aligns with industry's pragmatic stance of preserving communication performance, while enabling sensing opportunistically when feasible. While tradeoff functions can balance the importance of functionalities by using weights \cite{balef2023:piecewise-stationary-multi-objective-multi-armed-bandit-application-joint-communications-sensing}, changes in parameter settings (e.g., number of users, transmit power) can skew objective function values, rendering preset weights ineffective and shifting the intended operating point. To address this, we propose establishing strict hierarchies through careful weight design. Our approach consistently prioritizes communications regardless of parameter settings, ensuring its full optimization before addressing the sensing requirements, thus leading to a strictly tiered resource allocation framework.

% Figure 1
\begin{figure}[!t]
	\centering
	\includegraphics[width=0.62\columnwidth]{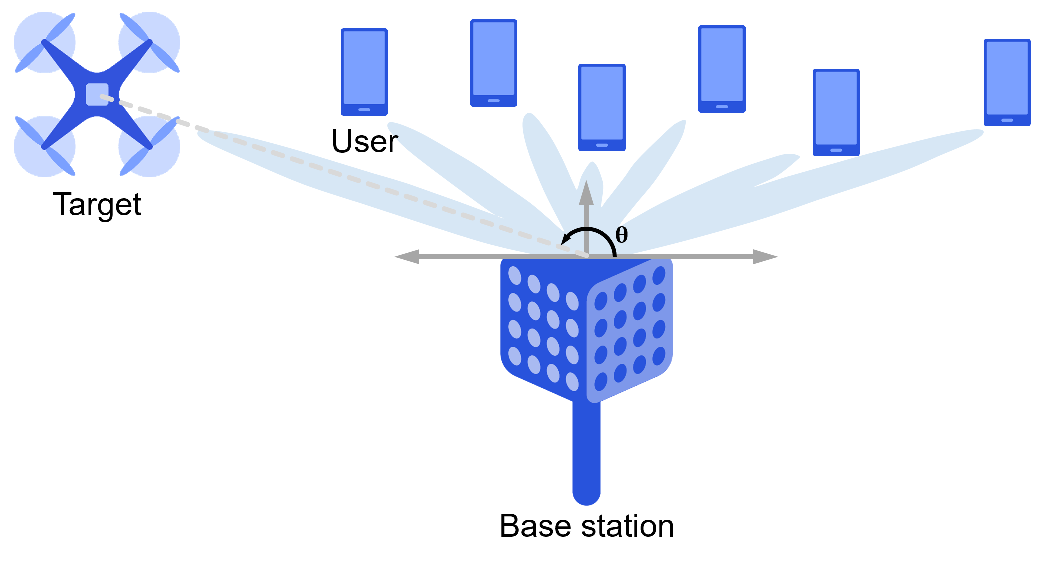}
	\vspace{-2mm}
	\caption{Multicast ISAC system with many users and a target.}	
	\label{fig:system-model}
	\vspace{-4mm}
\end{figure}

% Motivated by the preceding discussion, we investigate the joint design of admission control and multicast beamforming with constant-modulus discrete phases for \gls{ISAC} systems, where communications is prioritized and sensing is opportunistic. This results in a novel resource allocation problem whose key differences with respect to the existing state of the art are summarized in \cref{table:related-literature}. The resource allocation problem is formulated as a nonconvex \gls{MINLP}, which is challenging to solve. We propose an approach to reformulate it, leading to a \gls{MILP} that can be solved globally optimally. The proposed reformulation tackles the difficulties of solving the nonconvex \gls{MINLP} by implementing a series of transformations that convexify the nonconvex \gls{MINLP} without affecting optimality.

Motivated by the above discussion, we investigate the joint optimization of admission control and multicast beamforming with discrete phases for ISAC systems, prioritizing communications while enabling opportunistic sensing, and accounting for target angular uncertainty. This novel resource allocation problem, distinct from existing works (see \cref{table:related-literature}), is formulated as a nonconvex \gls{MINLP}, which is challenging to solve. We propose an approach to reformulate it, leading to a \gls{MILP} that can be solved globally optimally. Our approach employs a series of transformations to convexify the nonconvex MINLP without compromising optimality, effectively addressing the original problem's complexity. Additionally, we implement three baselines based on well-established optimization methods used in the resource allocation literature.

% Motivated by the preceding discussion, we investigate the joint optimization of admission control and multicast beamforming with constant-modulus discrete phases for \gls{ISAC} systems, where communications is prioritized and sensing is opportunistic.

\emph{Notation}: Boldface capital letters $ \mathbf{A} $ and boldface lowercase letters $ \mathbf{a} $ denote matrices and vectors, respectively. The transpose, Hermitian transpose, and trace of $ \mathbf{A} $ are denoted by $ \mathbf{A}^\mathrm{T} $, $ \mathbf{A}^\mathrm{H} $, and $ \mathrm{Tr} \left(  \mathbf{A} \right) $, respectively. The $l$-th row and $i$-th column of $ \mathbf{A} $ are denoted by $ \left[ \mathbf{A} \right]_{l,:} $ and $ \left[ \mathbf{A} \right]_{:,i} $, respectively, and the $l$-th element of $ \mathbf{a} $ is denoted by $ \left[ \mathbf{a} \right]_l $. $ \mathbb{C}^{I \times J} $ and $ \mathbb{N} $ denote the space of $ I \times J $ complex-valued matrices and the natural numbers, respectively. Also, $ j \triangleq \sqrt{-1} $ is the imaginary unit, $ \mathbb{E} \left\lbrace \cdot \right\rbrace  $ denotes statistical expectation, and $ \mathcal{CN} \left( \upsilon, \xi^2 \right) $ represents the complex Gaussian distribution with mean $ \upsilon $ and variance $ \xi^2 $.

%Also, $ j \triangleq \sqrt{-1} $ is the imaginary unit, $ \mathbb{E} \left\lbrace \cdot \right\rbrace  $ denotes statistical expectation, and $ \mathcal{CN} \left( \upsilon, \xi^2 \right) $ represents the complex Gaussian distribution with mean $ \upsilon $ and variance $ \xi^2 $.

%Also, $ j \triangleq \sqrt{-1} $ is the imaginary unit, $ \mathbb{E} \left\lbrace \cdot \right\rbrace  $ denotes statistical expectation, and $ \mathcal{CN} \left( \upsilon, \xi^2 \right) $ represents the complex Gaussian distribution with mean $ \upsilon $ and variance $ \xi^2 $.

% Table I
\begin{table}[!t]
	\begin{center}
	\fontsize{6.5}{6}\selectfont
	\setlength\tabcolsep{2.9pt}
	\centering
	\caption{Categorization of related work.}
	\label{table:related-literature}
	\begin{tabular}{c c c c c c c c} % Constraints
	\toprule
	\makecell{Works} & \makecell{$ \mathrm{D}_1 $}  & \makecell{$ \mathrm{D}_2 $} & \makecell{$ \mathrm{D}_3 $} &  \makecell{$ \mathrm{D}_4 $} & \makecell{$ \mathrm{D}_5 $} & \makecell{$ \mathrm{D}_6 $} & \makecell{$ \mathrm{D}_7 $} 
	\\
	\midrule
	\cite{cao2022:pareto-optimal-analog-beamforming-design-integrated-mimo-radar-communication, nguyen2023:multiuser-mimo-wideband-joint-communications-sensing-system-subcarrier-allocation, lyu2024:crb-minimization-ris-aided-mmwave-integrated-sensing-communications} & ISAC & \xmark & Continuous & Unicast & \xmark & \xmark & \xmark
	\\
	\midrule
	\cite{xu2023:joint-antenna-selection-beamforming-integrated-automotive-radar-sensing-communications-quantized-double-phase-shifters} & ISAC & \xmark & Discrete & Unicast & \xmark & \xmark & \xmark
	\\
	\midrule
	\cite{demir2015:optimum-discrete-phase-only-multicast-beamforming-joint-antenna-user-selection-cognitive-radio-networks} & Non-ISAC & \cmark & Discrete & Multicast & \xmark & \xmark & \xmark
	\\
	\midrule
	\cite{ren2024:fundamental-crb-rate-tradeoff-multi-antenna-isac-systems-information-multicasting-multi-target-sensing} & ISAC & \xmark & --- & Multicast & \xmark & \xmark & \xmark
	\\
	\midrule
	\cite{matskani2009:efficient-batch-adaptive-approximation-algorithms-joint-multicast-beamforming-admission-control} & Non-ISAC & \xmark & Continuous & Multicast & \cmark & \xmark & \xmark
	\\
	\midrule
	\cite{xu2018:robust-resource-allocation-uav-systems-uav-jittering-user-location-uncertainty} & ISAC & \xmark & Continuous & Unicast & \xmark & \cmark & \xmark
	\\
	\midrule
	\textbf{Proposed} & \textbf{ISAC} & \cmark & \textbf{Discrete} & \textbf{Multicast} & \bcmark & \bcmark & \bcmark
	\\
	\bottomrule
	\multicolumn{2}{l}{\tiny $ \mathrm{D}_1 $: System type} 
	&  
	\multicolumn{2}{l}{\tiny $ \mathrm{D}_2 $: Globally optimality}
	&  
	\multicolumn{2}{l}{\tiny $ \mathrm{D}_3 $: Phase type}
	&  
	\multicolumn{2}{l}{\tiny $ \mathrm{D}_4 $: Network topology}
	\\
	\multicolumn{2}{l}{\tiny $ \mathrm{D}_5 $: Admission control} 
	&
	\multicolumn{2}{l}{\tiny $ \mathrm{D}_6 $: Angle uncertainty} 
	&
	\multicolumn{2}{l}{\tiny $ \mathrm{D}_7 $: Hierarchical prioritization}
	\end{tabular}
	\end{center}
	\vspace{-4mm}
\end{table}

% !TeX root = ../main.tex

\section{System Model and Problem Formulation} \label{sec:system-model-problem-formulation}

We consider an \gls{ISAC} system comprising a \gls{BS} equipped with $ N $ transmit and $ N $ receive antennas, $ U $ single-antenna users, and one target, as shown in \cref{fig:system-model}.

\textbf{Beamforming:} The \gls{BS} transmits signal $ \mathbf{d} = \mathbf{w} z $, where $ \mathbf{w} \in \mathbb{C}^{N \times 1} $ is the multicast beamforming vector and $ z \in \mathbb{C} $ is the data symbol which serves both sensing and communication purposes simultaneously, and follows a complex Gaussian distribution with zero mean and unit variance, e.g., $ \mathbb{E} \left\lbrace z z^{*} \right\rbrace = 1 $. To account for the constant-modulus discrete phases used in the analog beamforming design, we include constraint $ {\mathrm{C}}_{1}: \left[ \mathbf{w} \right]_n \in \mathcal{S}, \forall n \in \mathcal{N} $, where $ \mathcal{N} = \left\lbrace 1, \dots, N \right\rbrace $ indexes the antenna elements and $ \mathcal{S} = \big\{ \delta e^{j \phi_1},  \dots, \delta e^{j \phi_L} \big\} $ is the set of admissible phases. In addition, $ \delta = \sqrt{P_\mathrm{tx}/N} $ is the magnitude, $ L $ is the number of phases, $ \phi_l $ is the $ l $-th phase, and $ P_\mathrm{tx} $ is the \gls{BS}'s transmit power. Furthermore, $ Q $ is the number of bits needed for encoding the $ L $ phases, i.e., $ Q = \log_2 (L) $.

\textbf{Admission control:} To decide which users are served by the \gls{BS}, we include constraint $ {\mathrm{C}}_{2}: \mu_u \in \left\lbrace 0, 1 \right\rbrace, \forall u \in \mathcal{U} $, where $ \mathcal{U} = \left\lbrace 1, \dots, U \right\rbrace $ indexes the users. Here, $ \mu_u = 1 $ indicates that user $ u $ is admitted, and $ \mu_u = 0 $ otherwise.

\textbf{Communications model:}
The signal received by user $ u $ is  $ y_{\mathrm{com},u} = \mathbf{h}^\mathrm{H}_u \mathbf{d} + \eta_{\mathrm{com},u} = \mathbf{h}^\mathrm{H}_u \mathbf{w} z + \eta_{\mathrm{com},u} $, where $ \mathbf{h}_u \in \mathbb{C}^{N \times 1} $ is the channel between the \gls{BS} and user $ u $, and $ \eta_{\mathrm{com},u} \sim \mathcal{CN} \left( 0,\sigma_\mathrm{com}^2 \right) $ is \gls{AWGN}. The communication \gls{SNR} at user $ u $ is
% Equation 2
\begin{equation}
	\mathsf{SNR}_{\mathrm{com},u} \left( \mathbf{w} \right) = \mathbf{w}^\mathrm{H} \widetilde{\mathbf{H}}_u \mathbf{w}, \forall u \in \mathcal{U},
\end{equation}
where $ \widetilde{\mathbf{H}}_u = \frac{ \mathbf{h}_u \mathbf{h}_u^\mathrm{H} }{\sigma_\mathrm{com}^2} $. Let $ \Gamma_\mathrm{th} $ be the minimum \gls{SNR} threshold necessary for successfully decoding the multicast data. To enforce this requirement jointly with user admission, we incorporate constraint $ \mathrm{C}_{3}: \mathbf{w}^\mathrm{H} \widetilde{\mathbf{H}}_u \mathbf{w} \geq \mu_u \cdot \Gamma_\mathrm{th}, \forall u \in \mathcal{U} $, i.e., the \gls{SNR} threshold must be satisfied for all admitted users.

\textbf{Sensing model:}
We assume the target is far from the \gls{BS}, thus we model it as a single point. 
The \gls{BS} operates as a monostatic co-located radar, i.e., the \gls{AoD} and \gls{AoA} are the same. Hence, the response matrix between the BS and the target is given by $ \mathbf{G} \left( \theta \right)  = \alpha \mathbf{a} \left( \theta \right) \mathbf{a}^\mathrm{H} \left( \theta \right) $, where $ \alpha $ is the reflection coefficient, $ \theta $ is the \gls{AoD}/\gls{AoA} of the target, and $ \mathbf{a} \left( \theta \right) = \left[ e^{j \pi \frac{-N+1}{2} \cos \left( \theta \right)}, \dots, e^{j \pi \frac{N-1}{2}\cos \left( \theta \right)} \right] ^\mathrm{T}, \in \mathbb{C}^{N \times 1} $ is the half-wavelength steering vector in the direction of $ \theta $. The reflected signal by the target at the \gls{BS} is $ y_\mathrm{sen} = \mathbf{w}^\mathrm{H} \mathbf{G} \left( \theta \right) \mathbf{d} + \eta_\mathrm{sen} = \mathbf{w}^\mathrm{H} \mathbf{G} \left( \theta \right) \mathbf{w} z + \eta_\mathrm{sen} $, where $ \eta_\mathrm{sen} \sim \mathcal{CN} \left( 0,\sigma_\mathrm{sen}^2 \right) $, Thus, the sensing \gls{SNR}, measured at he \gls{BS}, is given by
% Equation 5
\begin{equation} \label{eqn:dpg}
	\mathsf{SNR}_\mathrm{sen} \left( \mathbf{w}, \theta \right) = \mathbf{w}^\mathrm{H} \widetilde{\mathbf{G}} \left( \theta \right) \mathbf{w}.
\end{equation}
where $ \widetilde{\mathbf{G}} \left( \theta \right) = \frac{\mathbf{G} \left( \theta \right) }{\sigma_\mathrm{sen}^2} $. It is assumed that the transmit and receive antenna arrays are adequately spaced to prevent self-interference \cite{xiao2017:full-duplex-millimeter-wave-communication}. To account for potential uncertainty in the value of $ \theta $, e.g., caused by the target's speed  \cite{liu2022:integrated-sensing-communications-toward-dual-functional-wireless-networks-6g-beyond}, we adopt the model in \cite{xu2018:robust-resource-allocation-uav-systems-uav-jittering-user-location-uncertainty}, where an angular interval $ \left[ \theta - \Delta, \theta + \Delta \right] $ is considered, with $ \Delta $ representing the uncertainty in $ \theta $. This interval is discretized into samples, resulting in set $ \Theta = \left\lbrace \bar{\theta} \mid \bar{\theta} = \theta - \Delta + \tfrac{2 \Delta}{C-1} c \right\rbrace $, $ \forall c = 0, \dots, C -1 $, where $ C $ is the number of samples taken within the interval. To ensure that the sensing \gls{SNR} in all angular directions within $ \Theta $ exceeds some value $ \tau $, we first include constraint $ {\mathrm{C}}_{4}: \tau \geq 0 $  and then add constraint $ {\mathrm{C}}_{5}: \mathbf{w}^\mathrm{H} \widetilde{\mathbf{G}} \left( \theta \right) \mathbf{w} \geq \tau, \forall \theta \in \Theta $.

\textbf{Objective function:} We define the tradeoff function
% Equation 
\begin{equation} 
	f \left( \boldsymbol{\mu}, \tau \right) \triangleq \rho_\mathrm{com} \cdot f_\mathrm{com} \left( \boldsymbol{\mu} \right) + \rho_\mathrm{sen} \cdot f_\mathrm{sen} \left( \tau \right), 
\end{equation}
which we aim to maximize. Here, $ f_\mathrm{com} \left( \boldsymbol{\mu} \right) \triangleq \mathbf{1}^\mathrm{T} \boldsymbol{\mu} $ and $ f_\mathrm{sen} \left( \tau \right) \triangleq \tau $ are the objective functions related to communications and sensing, respectively. In particular, $ f_\mathrm{com} \left( \boldsymbol{\mu} \right) $ represents the number of admitted users, i.e., users that are served with the desired multicast data, while $ f_\mathrm{sen} \left( \tau \right) $ is the lowest sensing \gls{SNR} value for the angles in $ \Theta $. In addition, $ \boldsymbol{\mu} = \left[ \mu_1, \dots, \mu_U \right]^\mathrm{T} $, whereas $ \rho_\mathrm{com} $ and $ \rho_\mathrm{sen} $ are the weights that control the functionality importance.

\textbf{Problem formulation:}
 We formulate the joint design of admission control and discrete-phase beamforming as
% Problem P
\begin{align*} 
	% Objective
	\mathcal{P}: & \underset{\mathbf{w}, \boldsymbol{\mu}, \tau }{\mathrm{~maximize}}
	& f \left( \boldsymbol{\mu}, \tau \right) ~~~ \mathrm{s.t.} ~~~ {\mathrm{C}}_{1}, {\mathrm{C}}_{2}, {\mathrm{C}}_{3}, {\mathrm{C}}_{4}, {\mathrm{C}}_{5}.
\end{align*}

As a particular case, we consider that communications has higher hierarchy than sensing, achieved through the careful design of weights, as outlined in \cref{lem:lemma-1}. We highlight that our framework can accommodate any arbitrary weights, even when hierarchies are not required.
% Lemma 1
\begin{lemma} \label{lem:lemma-1}
	A set of weights ensuring that communications has higher hierarchy is given by $ \rho_\mathrm{com} = 1 $ and $ \rho_\mathrm{sen} = \tfrac{\sigma_\mathrm{sen}^2}{2 \alpha N P_\mathrm{tx} } $.
	%\vspace{-1mm}
\end{lemma}
\begin{proof}
	%  and promoting the prioritization of $ f_\mathrm{com} \left( \mathbf{w} \right) $ in a tiered manner.
	To ensure that the functionalities have different hierarchies, we choose the weights such that $ \rho_\mathrm{com} \cdot f_\mathrm{com} \left( \boldsymbol{\mu} \right) $ and $ \rho_\mathrm{sen} \cdot f_\mathrm{sen} \left( \tau \right) $ span nonoverlapping intervals. In particular, we let $ \rho_\mathrm{com} \cdot f_\mathrm{com} \left( \boldsymbol{\mu} \right) \in  \mathbb{N} $ handle the integer part of $ f \left( \boldsymbol{\mu}, \tau \right) $ and let $ \rho_\mathrm{sen} \cdot f_\mathrm{sen} \left( \tau \right) \in [ 0, 1) $ handle the decimal part, thereby effectively assigning a higher hierarchy to communications. Note that $ f_\mathrm{com} \left( \boldsymbol{\mu} \right) $ is an integer by definition. In order for $ \rho_\mathrm{com} \cdot  f_\mathrm{com} \left( \boldsymbol{\mu} \right) $ to also be an integer, we can choose any $ \rho_\mathrm{com} \in [ 1, \infty ) \cap \mathbb{N} $. For simplicity, we adopt $ \rho_\mathrm{com} = 1 $. Besides, since $ \tau $ is smaller than or equal to $ \mathbf{w}^\mathrm{H} \widetilde{\mathbf{G}} \left( \theta \right) \mathbf{w} $, $ \forall \theta \in \Theta $, as  stated in $ {\mathrm{C}}_{5} $, we can establish an upper bound for $ \tau $ by finding an upper bound for $ \mathbf{w}^\mathrm{H} \widetilde{\mathbf{G}} \left( \theta \right) \mathbf{w} $. Thus, we have that $ \mathbf{w}^\mathrm{H} \widetilde{\mathbf{G}} \left( \theta \right) \mathbf{w} = \tfrac{\alpha}{\sigma_\mathrm{sen}^2} \left| \mathbf{w}^\mathrm{H} \mathbf{a} \left( \theta \right) \right|^2 $. Applying the Cauchy-Schwarz inequality to the \gls{RHS}, we obtain that $ \tfrac{\alpha}{\sigma_\mathrm{sen}^2}  \left| \mathbf{w}^\mathrm{H} \mathbf{a} \left( \theta \right) \right|^2 \leq \tfrac{\alpha}{\sigma_\mathrm{sen}^2} \left\| \mathbf{a} \left( \theta \right) \right\|^2_2 \left\| \mathbf{w} \right\|^2_2 $. In addition, we have $ \left\| \mathbf{a} \left( \theta \right) \right\|^2_2 = N $ and $ \left\| \mathbf{w} \right\|^2_2 = P_\mathrm{tx} $, yielding $ \mathbf{w}^\mathrm{H} \widetilde{\mathbf{G}} \left( \theta \right) \mathbf{w} \leq \tfrac{\alpha  N P_\mathrm{tx}}{\sigma_\mathrm{sen}^2} $ and $ f_\mathrm{sen} \left( \tau \right) \leq \tfrac{\alpha  N P_\mathrm{tx}}{\sigma_\mathrm{sen}^2} $. In order for $ \rho_\mathrm{sen} \cdot f_\mathrm{sen} \left( \tau \right) \in [ 0, 1) $ to be true, we can choose any $ \rho_\mathrm{sen} \in \big( 0, \tfrac{\sigma_\mathrm{sen}^2}{\alpha N P_\mathrm{tx}} \big) $. For simplicity, we adopt $ \rho_\mathrm{sen} = \tfrac{\sigma_\mathrm{sen}^2}{2 \alpha N P_\mathrm{tx}} $.
\end{proof}

\section{Proposed Optimal Approach} \label{sec:proposed-approach}

We introduce constraint $ \mathrm{C}_{6}: \mathbf{W} = \mathbf{w} \mathbf{w}^\mathrm{H} $, which contains new variable $ \mathbf{W} $ that is related to $ \mathbf{w} $. This allows us to equivalently recast problem $ \mathcal{P} $ as problem $ \mathcal{P}' $, shown below,
\begin{align*} 
	% Objective
	\mathcal{P}': & \underset{ \mathbf{W}, \mathbf{w}, \boldsymbol{\mu}, \tau }{\mathrm{~maximize}}
	& & f \left( \boldsymbol{\mu}, \tau \right)
	\\
	% Constraint C1, C2, C4
	& ~~~~~ \mathrm{s.t.} & & \mathrm{C}_{1}, \mathrm{C}_{2}, \mathrm{C}_{4}, 
	\\
	% Constraint C3
	& & & \mathrm{C}_{3}: ~ \mathrm{Tr} \big( \widetilde{\mathbf{H}}_u \mathbf{W} \big) \geq \mu_u \cdot \Gamma_\mathrm{th}, \forall u \in \mathcal{U}, 
	\\
	% Constraint C5
	& & & \mathrm{C}_{5}: ~ \mathrm{Tr} \big( \widetilde{\mathbf{G}} \left( \theta \right) \mathbf{W} \big) \geq \tau, \forall \theta \in \Theta,
	\\
	% Constraint C6
	& & & \mathrm{C}_{6}: ~ \mathbf{W} = \mathbf{w} \mathbf{w}^\mathrm{H},
\end{align*}
where the trace commutative property was used in $ \mathrm{C}_{3} $ and $ \mathrm{C}_{5} $. Although constraint $ \mathrm{C}_{6} $ defines $ \mathbf{W} $ as rank-one, i.e., nonconvex and challenging to address, in our case $ \mathbf{W} $ arises from the product of discrete phases, encoded via binary variables. Particularly, the product of these variables can be optimally handled, as detailed in \cref{thm:proposition-1} to \cref{thm:proposition-5}.

% Proposition 1
\begin{proposition} \label{thm:proposition-1}
	Constraint $ \mathrm{C}_{1} $ can be equivalently rewritten as constraints $ \mathrm{D}_{1} $, $ \mathrm{D}_{2} $, and $ \mathrm{D}_{3} $, 
	% Constraints D1-D3
	\begin{equation} \nonumber
		\mathrm{C}_{1} \Leftrightarrow
			\begin{cases}
				   	\mathrm{D}_{1}: \left[ \mathbf{x}_n \right]_l \in \left\lbrace 0, 1\right\rbrace, \forall n \in \mathcal{N}, l \in \mathcal{L}, 
				   	\\	
				   	\mathrm{D}_{2}: \mathbf{1}^\mathrm{T} \mathbf{x}_n = 1, \forall n \in \mathcal{N},  
				   	\\
				   	\mathrm{D}_{3}: \left[ \mathbf{w} \right]_n = \mathbf{s}^\mathrm{T} \mathbf{x}_n, \forall n \in \mathcal{N}, 
			\end{cases}
	\end{equation}
	where vector $ \mathbf{s} \in \mathbb{C}^{L \times 1} $ is formed by all the elements in $ \mathcal{S} $, and $ \mathcal{L} = \left\lbrace 1, \dots, L \right\rbrace $.
\end{proposition} 
\begin{proof}%[Explanation]
	Binary variables can be used to encode the phase selection for each antenna. Hence, for each antenna $ n $, a binary vector $ \mathbf{x}_n $ is introduced, as stated in $ \mathrm{D}_{1} $. Also, only one phase per antenna must be selected, which is enforced by $ \mathrm{D}_{2} $.	Finally, $ \mathrm{D}_{3} $ maps $ \mathbf{x}_n $ to one of the phases in $ \mathbf{s} $.
\end{proof}

% Proposition 2
\begin{proposition} \label{thm:proposition-2}
	Constraint $ \mathrm{C}_{6} $ can be equivalently rewritten as constraint $ \mathrm{E}_{1} $, 
	% Constraints E1
	\begin{equation} \nonumber
		\mathrm{C}_{6} \Leftrightarrow \mathrm{E}_{1}: \left[ \mathbf{W} \right]_{n,m} = \left[ \mathbf{w} \right]_n \left[ \mathbf{w}^{*} \right]_m, \forall n, m \in \mathcal{N},
	\end{equation}
	where $ \left[ \mathbf{W} \right]_{n,m} $ represents the element of $ \mathbf{W} $ in the $ n $-th row and $ m $-th column.
\end{proposition} 
\begin{proof}%[Explanation]
	Given $ \mathbf{W} = \mathbf{w} \mathbf{w}^\mathrm{H} $, the $n$-th row of $ \mathbf{W} $ is $ \left[ \mathbf{w} \right]_n \mathbf{w}^\mathrm{H} $ and the $ m $-th element of that row is $\left[ \mathbf{w} \right]_n \left[ \mathbf{w}^{*} \right]_m $.
\end{proof}

% Proposition 3
\begin{proposition} \label{thm:proposition-3}
	Constraints $ \mathrm{D}_{3} $ and $ \mathrm{E}_{1} $ can be equivalently expressed as constraints $ \mathrm{F}_{1} $, $ \mathrm{F}_{2} $, and $ \mathrm{F}_{3} $,
	% Constraints F1-F3
	\begin{equation} \nonumber
		\mathrm{D}_{3}, \mathrm{E}_{1} \Leftrightarrow
			\begin{cases}
				   	\mathrm{F}_{1}: \left[ \mathbf{W} \right]_{n,m} = \mathrm{Tr} \left( \mathbf{S} \mathbf{x}_n \mathbf{x}^\mathrm{T}_m \right), \forall n \in \mathcal{N}, m \in \mathcal{M}_n,
				   	\\	
				   	\mathrm{F}_{2}: \left[ \mathbf{W} \right]_{m,n} = \left[ \mathbf{W} \right]^{*}_{n,m}, \forall n \in \mathcal{N}, m \in \mathcal{M}_n,
				   	\\	
				   	\mathrm{F}_{3}: \left[ \mathbf{W} \right]_{n,n} = \delta^2, \forall n \in \mathcal{N},  
			\end{cases}
	\end{equation}
	where $ \mathbf{S} = \mathbf{s}^{*} \mathbf{s}^\mathrm{T} $ and $ \mathcal{M}_n = \left\lbrace n+1, \dots, N \right\rbrace $.
\end{proposition} 
\begin{proof}%[Explanation]
	Replacing $ \mathrm{D}_{3} $ in $ \mathrm{E}_{1} $ leads to $ \left[ \mathbf{W} \right]_{n,m} =  \mathbf{s}^\mathrm{T} \mathbf{x}_n \mathbf{x}^\mathrm{T}_m \mathbf{s}^{*} = \mathrm{Tr} \left( \mathbf{S} \mathbf{x}_n \mathbf{x}^\mathrm{T}_m \right), \forall n, m \in \mathcal{N}. $ Since $ \mathrm{E}_{1} $ implies that $ \mathbf{W} $ is Hermitian, the elements with respect to the diagonal are conjugate symmetrical while the diagonal elements are $ \delta^2 $. Thus, instead of indexing the whole matrix $ \mathbf{W} $, we can only index the upper triangular part as accomplished by $ \mathrm{F}_{1} $, $ \mathrm{F}_{2} $, $ \mathrm{F}_{3} $, thereby reducing the number of decision variables.
\end{proof}

% Proposition 4
\begin{proposition} \label{thm:proposition-4}
	Constraint $ \mathrm{F}_{1} $ can be equivalently expressed as constraints $ \mathrm{G}_{1} $, $ \mathrm{G}_{2} $, and $ \mathrm{G}_{3} $, 
	% Constraints G1-G3
	\begin{equation} \nonumber
		\mathrm{F}_{1} \Leftrightarrow
			\begin{cases}
				   	\mathrm{G}_{1}: \mathbf{Y}_{n,m} = \mathbf{x}_n \mathbf{x}^T_m, \forall n \in \mathcal{N}, m \in \mathcal{M}_n,
				   	\\	
				   	\mathrm{G}_{2}: \left[ \mathbf{W} \right]_{n,m} = \mathrm{Tr} \left( \mathbf{S} \mathbf{Y}_{n,m}  \right), \forall n \in \mathcal{N}, m \in \mathcal{M}_n,  
				   	\\	
				   	\mathrm{G}_{3}: \left[ \mathbf{Y}_{n,m} \right]_{l,i} \in \left[ 0, 1 \right], \forall n \in \mathcal{N}, m \in \mathcal{M}_n, l, i \in \mathcal{L}, 
			\end{cases}
	\end{equation}
\end{proposition} 
\begin{proof}%[Explanation]
	Variable $ \mathbf{Y}_{n,m} $ is introduced to replace $ \mathbf{x}_n \mathbf{x}^T_m $, as shown in $ \mathrm{G}_{1} $. Applying $ \mathrm{G}_{1} $ to $ \mathrm{F}_{1} $ leads to $ \mathrm{G}_{2} $. Finally, $ \mathbf{Y}_{n,m} $ is defined as having entries in the interval $ \left[ 0, 1 \right] $, as stated by $ \mathrm{G}_{3} $. Here, it is not necessary to define $ \left[ \mathbf{Y}_{n,m} \right]_{l,i} $ as binary since the integrality nature is automatically enforced by $ \mathrm{G}_{1} $.
\end{proof}

% Proposition 5
\begin{proposition} \label{thm:proposition-5}
	Constraint $ \mathrm{G}_{1} $ can be equivalently expressed as constraints $ \mathrm{H}_{1} $ and $ \mathrm{H}_{2} $, 
	% Constraints G1-G3
	\begin{equation} \nonumber
		\mathrm{G}_{1} \Leftrightarrow
			\begin{cases}
				   	\mathrm{H}_{1}: \mathbf{1}^\mathrm{T} \left[ \mathbf{Y}_{n,m} \right]_{:,i} = \left[ \mathbf{x}_m \right]_{i}, \forall n \in \mathcal{N}, m \in \mathcal{M}_n, i \in \mathcal{L},
				   	\\	
				   	\mathrm{H}_{2}: \left[ \mathbf{Y}_{n,m} \right]_{l,:} \mathbf{1} = \left[ \mathbf{x}_n \right]_l, \forall n \in \mathcal{N}, m \in \mathcal{M}_n, l \in \mathcal{L},
			\end{cases}
	\end{equation}
\end{proposition} 
\begin{proof}%[Explanation]
	Due to space constraints, the following explanation focuses solely on $ \mathrm{H}_{2} $ and $ \mathrm{G}_{1} $. However, the relation between $ \mathrm{H}_{1} $ and $ \mathrm{G}_{1} $ is analogous, given the structural similarity between $ \mathrm{H}_{1} $ and $ \mathrm{H}_{2} $. Matrix $ \mathbf{Y}_{n,m} $ is the product of binary vectors $ \mathbf{x}_n $ and $ \mathbf{x}^T_m $, each of which has only one element $ 1 $. Hence, $ \mathbf{Y}_{n,m} $ also has one element $ 1 $ and the rest are $ 0 $. Therefore, if $ \left[ \mathbf{x}_n \right]_l = 1 $ and $ \left[ \mathbf{x}_m \right]_i = 1 $, then $ \left[ \mathbf{Y}_{n,m} \right]_{l,i} = 1 $. We  use this notion to get rid of the product $ \mathbf{x}_n \mathbf{x}^T_m $. According to $ \mathrm{G}_{1} $, $ \left[ \mathbf{x}_n \right]_l \mathbf{x}^T_m $ represents the $ l $-th row of $ \mathbf{Y}_{n,m} $, i.e., $ \left[ \mathbf{Y}_{n,m} \right]_{l,:} = \left[ \mathbf{x}_n \right]_l \mathbf{x}^T_m $. When $ \left[ \mathbf{x}_n \right]_l = 0 $, all the elements of $ \left[ \mathbf{Y}_{n,m} \right]_{l,:} $ are $ 0 $, which is equivalent to stating that the sum of all the elements of $ \left[ \mathbf{Y}_{n,m} \right]_{l,:} $ is $ 0 $. When $ \left[ \mathbf{x}_n \right]_l = 1 $, then $ \left[ \mathbf{Y}_{n,m} \right]_{l,:} = \mathbf{x}^T_m $. Since $ \mathbf{x}^T_m $ has one element equal to $ 1 $, then the sum of elements of $ \left[ \mathbf{Y}_{n,m} \right]_{l,:} $ must be $ 1 $. 
\end{proof}

After applying the transformation procedures above, problem $ \mathcal{P}' $ is equivalently recast as \gls{MILP} $ \mathcal{P}'' $
% Problem P''
\begin{align*} 
	% Objective
	\mathcal{P}'': & \underset{ \substack{ \mathbf{W}, \mathbf{X}, \mathbf{Y}, \boldsymbol{\mu}, \tau } } {\mathrm{~maximize}}
	& & f \left( \boldsymbol{\mu}, \tau \right) ~~~~~ \mathrm{s.t.} & & \substack{ \mathrm{C}_{2}, \mathrm{C}_{3}, \mathrm{C}_{4}, \mathrm{C}_{5}, \mathrm{D}_{1}, \mathrm{D}_{2}, \\ \mathrm{F}_{2}, \mathrm{F}_{3}, \mathrm{G}_{2}, \mathrm{G}_{3}, \mathrm{H}_{1}, \mathrm{H}_{2}.}
\end{align*}

Here, $ \mathbf{W} $ is not a variable but a placeholder and, therefore, it does not have to be defined as the former.

%\noindent \textsc{Remark:} \emph{
\begin{remark}
The worst-case computational complexity of problem $ \mathcal{P} $ is an \gls{ES}, which requires evaluating $ \mathcal{C}_\mathrm{ES} = 2^{Q N} \sum_{i = 0}^U {U \choose U-i} $ candidate solutions. However, the structure of $ \mathcal{P}'' $ allows us to utilize \gls{BnC} techniques, implemented in commercial solvers, which can solve $ \mathcal{P}'' $ optimally at a small fraction of $ \mathcal{C}_\mathrm{ES} $. In particular, \gls{BnC} operates by pruning suboptimal and infeasible candidate solutions \cite{desrosiers2010:branch-price-cut-algorithms}, but this lies beyond the scope of this work.
\end{remark}
%}}

% However, it can be solved optimally and efficiently using \gls{BnC} techniques implemented in commercial solvers, which prune and remove suboptimal and infeasible candidate solutions, thereby reducing complexity.

% !TeX root = ../main.tex

\section{Simulation Results} \label{sec:simulation-results}

We evaluate our proposed approach, {\mytextsf{OPT}}, for various parameter settings. We consider the Rician fading channel model, which allows \gls{LoS} and \gls{NLoS} channel components with different contributions. Thus, the channel for user $ u $ is given by $ \mathbf{h}_u = \gamma_u \mathbf{v}_u $, where $ \gamma_u $ accounts for large-scale fading and $ \mathbf{v}_u = \sqrt{K / (K+1)} \mathbf{v}^\mathrm{LoS}_u + \sqrt{1/(K+1)} \mathbf{v}^\mathrm{NLoS}_u, \forall u \in \mathcal{U} $ is the normalized small-scale fading, with $ K $ being the Rician fading factor. The \gls{LoS} component is defined as $ \mathbf{v}^\mathrm{LoS}_u = \mathbf{a} \left( \beta_u \right) $, where $ \beta_u $ is the \gls{LoS} angle, and the \gls{NLoS} components are defined as $ \mathbf{v}^\mathrm{NLoS}_u \sim \mathcal{CN} \left( \mathbf{0}, \mathbf{I} \right) $. To compute large-scale fading, we use the \texttt{UMa} channel model \cite{3gpp:38.901}, i.e., $ \gamma_u = 28 + 22 \log_{10} (d_u) + 20 \log_{10} (f_\mathrm{c}) $ dB, where $ f_\mathrm{c} $ is the carrier frequency and $ d_u $ is the distance between the \gls{BS} and user $ u $. Unless specified otherwise, we consider $ f_\mathrm{c} = 71 $ GHz, $ P_\mathrm{tx} = 36 $ dBm, $ \sigma_\mathrm{com}^2 = \sigma_\mathrm{sen}^2 = -84 $ dBm, $ N = 10 $, $ U = 5 $, $ Q = 3 $, $ \Gamma_\mathrm{th} = 30 $, $ \Delta = 0 $, $ C = 33 $, $ \theta = 120 $, $ \beta_1 = 30 $, $ \beta_2 = 40 $, $ \beta_3 = 50 $, $ \beta_4 = 60 $, $ \beta_5 = 70 $, $ d_1 = \dots = d_5 = 40 $ m, $ \mathcal{S} = \big\{ \delta, \delta e^{j 2 \pi\frac{1}{2^Q}}, \dots, \delta e^{j 2 \pi\frac{2^Q-1}{2^Q}} \big\} $, and $ \alpha = \frac{\lambda^2 R}{64 \pi^3 \tilde{d}^4} $, where $ \lambda $ is the wavelength, $ R = 1 $ is the radar cross-section, $ \tilde{d} = 20 $ m is the distance between the \gls{BS} and the target \cite{wen2022:joint-secure-communication-radar-beamforming-secrecy-estimation-rate-based-design}. We used \texttt{CVX} and \texttt{MOSEK} on a laptop equipped with 16GB of RAM and an Intel Core i7@1.8GHz processor for our simulations.

% Scenario I:
\textbf{\emph{Scenario I:} } We investigate the impact of the number of antennas, $ N $, number of quantization bits, $ Q $, and transmit power, $ P_\mathrm{tx} $, on the sensing performance, $ f_\mathrm{sen} \left( \tau \right) $. \cref{fig:scenario-a} shows that as $ N $ increases, $ f_\mathrm{sen} \left( \tau \right) $ improves due to enhanced directivity. Similarly, higher $ P_\mathrm{tx} $ allows the \gls{SNR} threshold, $ \Gamma_\mathrm{th} $, to be met more easily, leaving additional power for sensing. Interestingly, increasing from $ Q = 3 $ to $ Q = 5 $ has minimal impact on $ f_\mathrm{sen} \left( \tau \right) $, with average differences within $ 5.3\% $ between $ Q = 3 $ and $ Q = 4 $, and within $ 6.5 \% $ between $ Q = 3 $ and $ Q = 5 $ when $ P_\mathrm{tx} = 42 $ dBm. The average runtime complexities for $ Q = 3 $, $ Q = 4 $, and $ Q = 5 $ are $ 0.3109 $, $ 0.9881 $, and $ 6.9847 $ s, respectively. Thus, in the sequel we use $ Q = 3 $. We highlight that with the current configuration, all users are served, i.e., $ f_\mathrm{com} \left( \boldsymbol{\mu} \right) = 5 $. In subsequent scenarios, we show the effect of user admission on the sensing performance. % Given that computational complexity increases with $ Q $, a smaller number is preferred. Thus, in the sequel, we adopt $ Q = 3 $ since this value offers a cost-effective trade-off. 

% Scenario II:
\textbf{\emph{Scenario II:} } We investigate the impact of the \gls{SNR} threshold, $ \Gamma_\mathrm{th} $, and target's angle uncertainty, $ \Delta $, on communications and sensing performance. \cref{fig:scenario-b1} shows that the number of admitted users increases as $ P_\mathrm{tx} $ grows. Besides, when $ P_\mathrm{tx} < 22 $ dBm, no user can be served, thus all available power is allocated to sensing. At $ P_\mathrm{tx} = 22 $ dBm, only one user is served, while at $ P_\mathrm{tx} = 28 $ dBm, all users are admitted. Beyond this point, further increases in $ P_\mathrm{tx} $ do not affect $ f_\mathrm{com} \left( \boldsymbol{\mu} \right) $ as it has reached its maximum, thus surplus power is used to enhance $ f_\mathrm{sen} \left( \tau \right) $. In \cref{fig:scenario-b2}, the system remains in sensing-only mode over a larger $ P_\mathrm{tx} $ range since $ \Gamma_\mathrm{th} $ increased from $ 30 $ to $ 60 $, demanding more power per admitted user. This explains why it takes until $ P_\mathrm{tx} = 26 $ dBm for the system to transition into sensing and communications mode. Additionally, $ P_\mathrm{tx} = 32 $ dBm is needed to serve all users, compared to \cref{fig:scenario-b1} where $ P_\mathrm{tx} = 28 $ dBm was sufficient. Lastly, \cref{fig:scenario-b3} illustrates the impact of increasing angle uncertainty $ \Delta $ from $ 0 $ to $ 8 $, as compared to Fig. 3b. Despite the larger $ \Delta $, the performance of admission control remains unaffected, while sensing performance declines. This is because admission control takes priority, meaning that sensing is only considered after the communication functionality has been fully optimized. A larger $ \Delta $ spreads surplus power over a wider angular range, reducing power concentration in the directions of interest, degrading sensing performance. As $ f_\mathrm{sen} \left( \tau \right) $ is small due to the target's reflection, we have scaled it by a factor of $ 100 $ in \cref{fig:scenario-b} for visualization purposes only.

% Figure: Scenario I
\begin{figure}[!t]
	\centering
	\includegraphics[width=0.90\columnwidth]{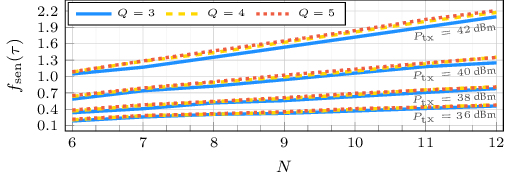}
	\vspace{-2.5mm}
	\caption{Impact of the number of antennas, transmit power, and phase resolution on the sensing performance.}	
	\label{fig:scenario-a}
	\vspace{-3mm}
\end{figure}
% Figure: Scenario II
\begin{figure}[!t]
	% Figure 1a
 	\begin{subfigure}[b]{0.32\columnwidth}
		\centering
		\includegraphics[height=2.8cm]{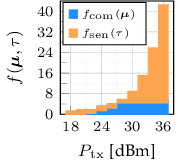}
		\caption{\footnotesize $ \Gamma_\mathrm{th} = 30 $, $ \Delta = 0 $.}	
		\label{fig:scenario-b1}
 	\end{subfigure}
    ~ 
	% Figure 1b
 	\begin{subfigure}[b]{0.3\columnwidth}
		\centering
		\includegraphics[height=2.8cm]{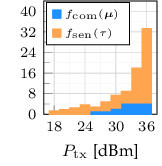}
		\caption{\footnotesize $ \Gamma_\mathrm{th} = 60 $, $ \Delta = 0 $.}	
		\label{fig:scenario-b2}
 	\end{subfigure}
    ~ 
	% Figure 1c
 	\begin{subfigure}[b]{0.31\columnwidth}
		\centering
		\includegraphics[height=2.8cm]{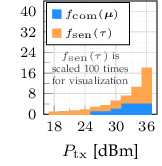}
		\caption{\footnotesize $ \Gamma_\mathrm{th} = 60 $, $ \Delta = 8 $.}	
		\label{fig:scenario-b3}
 	\end{subfigure}
    \caption{Impact of SNR threshold and target's angle uncertainty on sensing and communications.}
 	\label{fig:scenario-b}
 	\vspace{-3mm}
\end{figure}
% Figure: Scenario D1
\begin{figure}[!t]
	\centering
	\includegraphics[width=0.90\columnwidth]{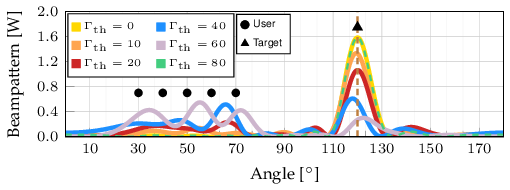}
	\vspace{-2.5mm}
	\caption{Impact of SNR threshold on the beampattern.}	
	\label{fig:scenario-c}
	\vspace{-3mm}
\end{figure}

% Scenario III:
\textbf{\emph{Scenario III:}} 
We show how the beampattern adapts to meet different \gls{ISAC} modes, considering $ P_\mathrm{tx} = 32 $ dBm. Here, we assume $ \mu_1 = \dots = \mu_5 $ for admission control, allowing us to observe sharper changes in the beampattern. \cref{fig:scenario-c} illustrates that as $ \Gamma_\mathrm{th} $ increases, the power towards the users grows, while the power illuminating the target decreases. When $ \Gamma_\mathrm{th} $ becomes too high to meet for the users, the \gls{ISAC} system collapses to sensing-only. That is observable with $ \Gamma_\mathrm{th} = 80 $, as the produced beampattern is identical to that obtained with $ \Gamma_\mathrm{th} = 0 $, demonstrating that our approach adapts to fully prioritize sensing when communications is not feasible.

\textbf{\emph{Scenario IV:}} We compare {\mytextsf{OPT}} against three baselines, {\mytextsf{BL1}}, {\mytextsf{BL2}}, and {\mytextsf{BL3}}, which we implemented using well-known optimization techniques. These are briefly described below.

\emph{Baseline 1 }({\mytextsf{BL1}}): This approach is based on \gls{SDR}, as in \cite{cao2022:pareto-optimal-analog-beamforming-design-integrated-mimo-radar-communication}. As the solution is not rank-one necessarily, the principal eigenvector is obtained via eigendecomposition. Afterwards, randomization and phase projection are applied to the eigenvector to meet the \gls{SNR} requirements. A total of $ 10^4 $ randomizations are employed.

\emph{Baseline 2 }({\mytextsf{BL2}}): This approach is based on the inner approximation of $ \mathrm{C}_3 $ and $ \mathrm{C}_5 $, as in \cite{abanto2023:radiorchestra-proactive-management-millimeter-wave-self-backhauled-small-cells-joint-optimization-beamforming-user-association-rate-selection-admission-control}, where a more conservative convex inequality is used to avoid dealing with the product $ \mathbf{w} \mathbf{w}^\mathrm{H} $. Thus, we use $ \mathfrak{Re} \left( \mathbf{h}_u^\mathrm{H} \mathbf{w} \right) \geq \sigma \mu_u \cdot \sqrt{\Gamma_\mathrm{th}} $ instead of $ \mathrm{C}_3: \mathrm{Tr} \big( \widetilde{\mathbf{H}}_u \mathbf{W} \big) \geq \mu_u \cdot \Gamma_\mathrm{th} $, and we do the same with $ \mathrm{C}_5 $.

\emph{Baseline 3 }({\mytextsf{BL3}}): This approach is based on \gls{SCA}, as in \cite{lyu2024:crb-minimization-ris-aided-mmwave-integrated-sensing-communications}. The obtained solution may not satisfy the set of discrete phases, thus the phases need to be projected. If the SNR constraints are not met after phase projection, $ 10^4 $ randomizations and projection are employed.

% Figure: Scenario E1
\begin{figure}[!t]
	\begin{center}
	\includegraphics[width=0.90\columnwidth]{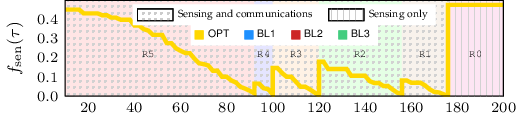}
	\includegraphics[width=0.90\columnwidth]{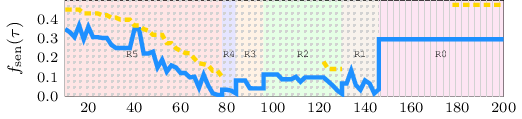}
	\includegraphics[width=0.90\columnwidth]{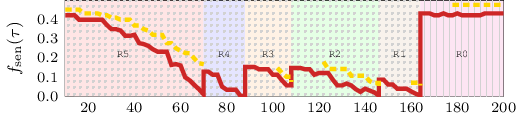}
	\includegraphics[width=0.90\columnwidth]{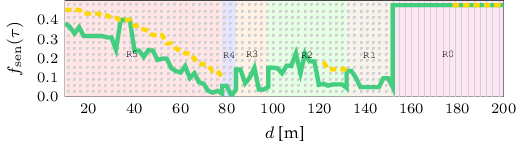}
	\end{center}
	\vspace{-1.5mm}
	\caption{Performance comparison of four different approaches.}	
	\label{fig:scenario-d}
\end{figure}

%We demonstrate how the proposed radio resource allocation enables the \gls{ISAC} system to seamlessly transition between sensing and communications mode and sensing-only mode, as the distance between the BS and users increases. \cref{fig:scenario-d} illustrates six distinct regions for each of the benchmarked approaches, labeled \texttt{RX}, where \texttt{X} represents the number of admitted users. For example, in \texttt{R5}, five users are served whereas, in \texttt{R0}, only sensing is supported. In this scenario, we assume that the distance between the BS and the users is $ d $. As $ d $ increases, the number of admitted users decreases due to higher path loss, which makes it increasingly difficult to meet the SNR threshold for all users. Since the communications functionality has higher hierarchy, wider areas $ \texttt{R1-R5} $ indicate a more efficient approach. Note that {\mytextsf{OPT}} achieves wider regions $ \texttt{R1-R5} $, while the baselines have narrower regions and transition between regions swifter than {\mytextsf{OPT}}, undermining the prioritization of communications. We have plotted {\mytextsf{OPT}} where the regions intersect, showing that {\mytextsf{OPT}} yields higher values with average gains of $ 59 \% $, $ 39\% $, and $ 47\% $ over {\mytextsf{BL1}}, {\mytextsf{BL2}} and {\mytextsf{BL3}}, respectively, further demonstrating its superiority.  

We demonstrate how the proposed radio resource allocation enables the \gls{ISAC} system to efficiently transition between joint sensing and communication and sensing-only modes as the distance between the BS and users increases. \cref{fig:scenario-d} illustrates six regions \texttt{RX}, where \texttt{X} represents the number of admitted users, e.g., in \texttt{R5}, five users are served, while in \texttt{R0}, only sensing is supported. In this scenario, the distance between the BS and the users is $ d $. As $ d $ increases, fewer users are admitted due to greater path loss, making it harder to meet the \gls{SNR} threshold. Since communications take priority, larger $ \texttt{R1-R5} $ regions indicate a more efficient approach. Notably, {\mytextsf{OPT}} achieves broader regions compared to the baselines, which transition between regions more rapidly, undermining communication prioritization. We plotted {\mytextsf{OPT}} where it intersected with the baselines in the same regions, and report that, within the distance range of $ \left[ 10, 66 \right] $ m, {\mytextsf{OPT}} yields higher values, with average gains of $ 59 \% $, $ 39\% $, and $ 47\% $ over {\mytextsf{BL1}}, {\mytextsf{BL2}}, and {\mytextsf{BL3}}, respectively, further demonstrating its superiority. The average runtime complexities for {\mytextsf{OPT}}, {\mytextsf{BL1}}, {\mytextsf{BL2}}, and {\mytextsf{BL3}}, are $ 0.3867 $, $ 0.2951 $, $ 0.1917 $, $ 0.2435 $s, respectively.

% !TeX root = ../main.tex

\section{Conclusions} \label{sec:conclusions}

%We investigated a novel resource allocation problem for \gls{ISAC} systems, focusing on the joint design of admission control and multicast beamforming with discrete phases, while considering different hierarchies for sensing and communications. Sensing performance was evaluated using the sensing \gls{SNR}, whereas communication performance was evaluated by the number of admitted multicast users. In this letter, communications was given higher priority. We proposed an approach that yields globally optimal solutions for this problem, demonstrating superior performance and offering wider operating regions for communications compared to three baselines.

We investigated a novel resource allocation problem in \gls{ISAC} systems, focusing on the joint design of admission control and multicast beamforming with discrete phases, while addressing different priority levels for sensing and communication. Communication was given higher priority, with its performance measured by the number of admitted users, while sensing performance was evaluated using the sensing \gls{SNR}. Our proposed approach achieves globally optimal solutions, outperforming three baseline methods and enabling broader communication operating regions.

%We investigated a novel resource allocation problem in \gls{ISAC} systems, focusing on the joint design of admission control and multicast beamforming with discrete phases, while addressing different priority levels for sensing and communication. Communication was prioritized, with sensing performance evaluated via the sensing \gls{SNR} and communication performance measured by the number of admitted multicast users. Our proposed approach offers globally optimal solutions, outperforming three baseline methods and providing broader communication operating regions.
% !TeX root = ../main.tex

\section*{Acknowledgment} \label{sec:acknowledgment}
This research was supported by the German Federal Ministry of Education and Research under ``Project 16KISK035''.

%\begin{spacing}{1.3}
\bibliographystyle{IEEEtran}
\bibliography{IEEEabrv,ref}

%\input{sections/appendices}
%\newpage

\end{document}